\documentclass[10pt]{article}
\usepackage[preprint]{tmlr}

\usepackage{amsmath,amssymb,amsfonts,amsthm,mathtools,bm}
\usepackage{booktabs}
\usepackage{enumitem}
\usepackage{microtype}
\usepackage{graphicx}
\usepackage{xcolor}
\usepackage{subcaption}
\usepackage{placeins}
\usepackage{url}
\usepackage{hyperref}
\usepackage[nameinlink,noabbrev]{cleveref}
\allowdisplaybreaks

\hypersetup{colorlinks=true,linkcolor=blue,citecolor=blue,urlcolor=blue}

\newenvironment{keywords}{\par\smallskip\noindent\textbf{Key words.}\ }{\par\smallskip}
\newenvironment{AMS}{\par\smallskip\noindent\textbf{MSC codes.}\ }{\par\smallskip}

\theoremstyle{plain}
\newtheorem{theorem}{Theorem}[section]
\newtheorem{lemma}[theorem]{Lemma}
\newtheorem{proposition}[theorem]{Proposition}
\newtheorem{corollary}[theorem]{Corollary}

\theoremstyle{definition}
\newtheorem{definition}[theorem]{Definition}
\newtheorem{algorithm}[theorem]{Algorithm}
\newtheorem{remark}[theorem]{Remark}

\newcommand{\R}{\mathbb R}

\newcommand{\HH}{\ensuremath{\mathbb{H}}}
\newcommand{\B}{\mathbb B}
\newcommand{\calB}{\mathcal B}

\newcommand{\calK}{\mathcal K}
\newcommand{\calM}{\mathcal M}
\newcommand{\dist}{\operatorname{dist}}
\newcommand{\vol}{\operatorname{vol}}
\newcommand{\argmin}{\operatorname*{argmin}}
\newcommand{\ip}[2]{\left\langle #1,#2\right\rangle}
\newcommand{\ipL}[2]{\left\langle #1,#2\right\rangle_{\rm L}}
\newcommand{\norm}[1]{\left\lVert #1\right\rVert}
\newcommand{\abs}[1]{\left\lvert #1\right\rvert}
\newcommand{\eps}{\varepsilon}
\newcommand{\kappaC}{\kappa}
\newcommand{\Id}{I}

\newcommand{\gbar}{\bar g}
\newcommand{\arcosh}{\operatorname{arcosh}}
\newcommand{\artanh}{\operatorname{artanh}}
\newcommand{\sech}{\operatorname{sech}}

\numberwithin{equation}{section}

\title{One-Shot Klein Cutting Planes for Lipschitz Geodesically Convex Optimization in Hyperbolic Space}

\author{
\name Yutong Zhang \email  \\
\addr School of Mathematics, Sichuan University, Chengdu 610064, China
\AND
\name Yaoran Yang \email  \\
\addr School of Mathematics, Sichuan University, Chengdu 610064, China
\AND
\name Yifan Zhu \email  \\
\addr College of Computer Science, Sichuan University, Chengdu 610064, China
\AND
\name Wentao Zhang \email zhang-wt24@mails.tsinghua.edu.cn \\
\addr Tsinghua Shenzhen International Graduate School, Tsinghua University, Shenzhen 518055, China\\
Corresponding author
}

\def\month{MM}
\def\year{YYYY}
\def\openreview{\url{https://openreview.net/forum?id=XXXX}}

\begin{document}
\maketitle

\begin{abstract}
Motivated by the COLT 2023 open problem of Criscitiello, Mart\'inez-Rubio, and Boumal on deterministic first-order methods for Lipschitz geodesically convex optimization on Hadamard manifolds, we study hyperbolic space
\[
\HH^d_{-\kappaC^2}
=\{X\in\R^{d+1}:\ipL{X}{X}=-1,\ X_0>0\},
\qquad
\ip{U}{V}_X=\kappaC^{-2}\ipL{U}{V}.
\]
For every geodesically convex $M$-Lipschitz function
\[
f:\bar B_{\HH}(x_0,r)\to\R,\qquad s=\kappaC r,
\]
we give a one-shot Klein cutting-plane method returning a queried point $\hat x$ such that
\[
f(\hat x)-\min_{\bar B_{\HH}(x_0,r)}f\le \eps Mr
\]
after at most
\[
\left\lceil
2d(d+1)\log\!\left(\frac{16\sinh s\cosh s}{s\eps}\right)
\right\rceil
\]
oracle calls. For $d\ge2$, each localization step costs $O(d^2)$ arithmetic operations; for $d=1$, an interval variant gives the same oracle bound. Hence
\[
N=O\bigl(d^2(s+\log(e/\eps))\bigr)
=O\bigl(d^2\zeta_s\log(e/\eps)\bigr),
\qquad
\zeta_s=s/\tanh s .
\]
Compared with the constant-curvature construction associated with the COLT problem, this replaces chained curvature--accuracy dependence by additive dependence. The proof does not rely on convexity of the Klein pullback, which is generally only quasiconvex. Instead, every Riemannian subgradient halfspace becomes an exact Euclidean central cut: for $\theta=\kappaC\dist(X,Y)$,
\[
\ip{g}{\log_XY}_X
=\frac{\theta}{\kappaC^2\sinh\theta}\ipL{g}{Y},
\]
and tangency at $X$ converts $\ipL{g}{Y}\le0$ into
\[
\gbar^{\mathsf T}(u-c)\le0,\qquad u=\Phi(Y),\ c=\Phi(X).
\]
Thus one fixed Euclidean ellipsoid localizes the hyperbolic ball, and curvature enters only through
\[
\log\!\left(\frac{\sinh s\cosh s}{s\eps}\right)
=\log(1/\eps)+2s-\log(4s)+O(e^{-4s}).
\]
The general Hadamard-manifold problem remains open.
\end{abstract}

\begin{keywords}
geodesically convex optimization, hyperbolic space, ellipsoid method, cutting planes, Beltrami--Klein model, first-order oracle, Riemannian subgradient
\end{keywords}

\begin{AMS}
90C25, 90C48, 65K05, 53B20, 52A41
\end{AMS}

\section{Motivation from the COLT open problem and the hyperbolic setting}

The Euclidean ellipsoid method separates optimization into two statements:
\[
        F(v)\ge F(u)+p^{\mathsf T}(v-u),
        \qquad
        \{v:F(v)\le F(u)\}\subseteq\{v:p^{\mathsf T}(v-u)\le0\},
\]
and
\[
        \vol(E^+)\le \exp\!\left(-\frac1{2(d+1)}\right)\vol(E),
        \qquad
        E\cap\{p^{\mathsf T}(v-c_E)\le0\}\subseteq E^+ .
\]
For the Euclidean problem
\[
        \min_{u\in\bar B(u_0,R)}F(u),
        \qquad
        F\text{ convex and }L\text{-Lipschitz},
\]
these two facts imply
\[
        N=O\!\left(d^2\log\frac{LR}{\eta}\right),
        \qquad
        F(\hat u)-F^*\le\eta .
\]

On a Riemannian manifold $(\calM,g)$ the first-order inequality is
\[
        f(y)\ge f(x)+\ip{\xi}{\log_x(y)}_x,
        \qquad
        \xi\in\partial f(x),
\]
so the canonical localization cut is
\[
        H(x,\xi)=\{y\in\calM:\ip{\xi}{\log_x(y)}_x\le0\}.
\]
The COLT open problem of Criscitiello, Mart\'inez-Rubio, and Boumal asks for deterministic first-order methods for Lipschitz geodesically convex optimization in the general Hadamard-manifold setting, with query complexity polynomial in the relevant dimension and curvature parameters and logarithmic in the target accuracy.  In the hyperbolic model-space benchmark, the corresponding parameters are
\[
        s=\kappaC r,
        \qquad
        \zeta_s=\frac{s}{\tanh s},
\]
and the desired accuracy scale is
\[
        f(\hat x)-f^*\le \eps Mr.
\]
The obstruction in the general Hadamard setting is that the family
\[
        H(x,\xi)=\{y\in\calM:\ip{\xi}{\log_x(y)}_x\le0\}
\]
is not usually an affine family in one fixed chart; a moving tangent-space localization scheme must then control both transport and volume distortion.  The constant-curvature construction described in \cite{CMB2023Open} uses local model-space subproblems and incurs an additional accuracy logarithm.  Our goal here is narrower: we sharpen this constant negative-curvature construction by exploiting the global projective structure of the Beltrami--Klein model.

The present paper proves the following theorem for hyperbolic space.  The statement is deliberately separated into the sharp additive bound and a coarser product-style bound comparable to the parameterization used in the COLT open-problem discussion, because
\[
        s+\log(1/\eps)
        \le (1+s)\log(e/\eps)
        \le 2\zeta_s\log(e/\eps),
        \qquad
        0<\eps<1,
\]
where the second inequality follows from
\[
        1+s\le 2\zeta_s,
        \qquad
        \zeta_s=\frac{s}{\tanh s},
        \qquad
        \lim_{s\downarrow0}\zeta_s=1.
\]

\begin{theorem}[global Klein cutting-plane theorem]\label{thm:main}
Let
\[
        d\ge1,
        \qquad
        \kappaC>0,
        \qquad
        r>0,
        \qquad
        s=\kappaC r,
        \qquad
        \eps\in(0,1).
\]
Let
\[
        f:\bar B_{\HH}(x_0,r)\to\R
\]
be geodesically convex and $M$-Lipschitz with respect to hyperbolic distance on the closed ball.  Suppose a first-order oracle returns
\[
        (f(x),g),
        \qquad
        g\in\partial f(x),
\]
at every queried point $x\in\bar B_{\HH}(x_0,r)$.  Then \Cref{alg:gkcp} returns a queried point $\hat x$ such that
\[
        f(\hat x)-\min_{x\in\bar B_{\HH}(x_0,r)}f(x)
        \le \eps Mr
\]
after no more than
\[
        N(d,s,\eps)=
        \left\lceil
        2d(d+1)
        \log\!\left(\frac{16\sinh s\cosh s}{s\eps}\right)
        \right\rceil
\]
first-order oracle calls.  For $d\ge2$ each localization update uses
\[
        Q a,
        \qquad
        a^{\mathsf T}Qa,
        \qquad
        bb^{\mathsf T},
\]
and therefore costs $O(d^2)$ arithmetic operations apart from the oracle and standard hyperboloid--Klein coordinate evaluations.  For $d=1$ the localization body is an interval and each update costs $O(1)$.  Moreover
\[
        N(d,s,\eps)
        \le
        \left\lceil
        2d(d+1)\left(2s+\log\frac{16}{\eps}\right)
        \right\rceil,
\]
so the sharp additive form and the coarser product-style form are
\[
        N(d,s,\eps)=O\bigl(d^2(s+\log(1/\eps))\bigr),
        \qquad
        N(d,s,\eps)=O\bigl(d^2\zeta_s\log(e/\eps)\bigr).
\]
\end{theorem}

\begin{remark}[scope and relation to the general Hadamard problem]\label{rem:scope}
The theorem is a hyperbolic-space result: it covers arbitrary dimension, arbitrary radius, arbitrary Lipschitz geodesically convex objectives on a hyperbolic ball, deterministic first-order access, and the accuracy scale $\eps Mr$.  It should not be read as a solution of the general Hadamard-manifold open problem.  In particular, it does not give a one-chart localization method for arbitrary Hadamard manifolds, nor for arbitrary manifolds with curvature in $[-K,K]$.  The proof uses two exact identities,
\[
        \Phi(\text{hyperbolic geodesic})=\text{Euclidean chord},
        \qquad
        \ip{g}{\log_XY}_{X}\le0
        \Longleftrightarrow
        a_g^{\mathsf T}(\Phi(Y)-\Phi(X))\le0,
\]
which are projective constant-curvature identities and are unavailable in general Hadamard geometry.
\end{remark}
\section{Related work}\label{sec:related-work}

The present work is motivated most directly by the COLT open problem of
Criscitiello, Mart\'inez-Rubio, and Boumal \cite{CMB2023Open}.  The general
open problem concerns deterministic first-order methods for Lipschitz
geodesically convex optimization on Hadamard manifolds, with polynomial
dependence on the relevant dimension and curvature parameters and logarithmic
dependence on the target accuracy.  In the hyperbolic model-space setting, the
natural scale is
\[
        s=\kappaC r,
        \qquad
        \zeta_s=\frac{s}{\tanh s},
\]
and the requested accuracy is of order $\eps Mr$ on a ball of radius $r$.
The constant-curvature construction described in that work proceeds through
local charts of controlled radius.  This gives the desired qualitative
dependence on curvature, but incurs an additional accuracy logarithm because
progress is chained over many local phases.  Our contribution is to remove this
chaining in the negative constant-curvature setting: the Beltrami--Klein chart
converts every Riemannian subgradient cut into a single Euclidean central cut,
so one global ellipsoid localization suffices.

The complexity landscape for geodesically convex optimization is constrained by
lower bounds depending on curvature.  Criscitiello and Boumal
\cite{CriscitielloBoumal2023Lower} show that curvature can genuinely increase
the oracle complexity of geodesically convex optimization, so a global
hyperbolic method must pay a curvature-dependent price.  The bound proved here
is consistent with this viewpoint: the dependence on $s$ is not eliminated, but
it appears additively in
\[
        \log\left(\frac{\sinh s\cosh s}{s\eps}\right)
        =
        \log(1/\eps)+2s-\log(4s)+O(e^{-4s}),
\]
rather than as a product of a curvature term and an accuracy logarithm.  Thus
the result should be read not as curvature-free optimization, but as a sharp
use of the projective structure of hyperbolic space to avoid unnecessary local
restarts.

Another closely related direction is Riemannian acceleration in model spaces.
Mart\'inez-Rubio \cite{MartinezRubio2022} studies global acceleration in
hyperbolic and spherical spaces, showing that constant-curvature geometry can
support algorithmic behavior stronger than what is available on arbitrary
manifolds.  The present paper is different in both oracle model and algorithmic
mechanism.  It addresses nonsmooth Lipschitz geodesically convex objectives
through first-order cutting planes, rather than smooth accelerated methods.
Nevertheless, both lines of work use the special structure of constant
curvature to obtain global algorithms whose guarantees would not follow from
generic local-coordinate arguments alone.

On the Euclidean side, our proof follows the classical separation between
first-order convexity and volume reduction that underlies the ellipsoid method.
The standard convex-optimization framework and the role of Lipschitz constants,
subgradients, and accuracy scaling are treated systematically by Nesterov
\cite{Nesterov2004}.  In Euclidean space, a subgradient inequality gives an
affine halfspace containing the lower sublevel set, and the ellipsoid method
then localizes a minimizer by repeated central cuts.  Our hyperbolic method
preserves exactly this logic, but not by proving that the Klein pullback of the
objective is Euclidean convex.  Instead, the pullback is generally only
quasiconvex, and the key observation is that the Riemannian lower-sublevel cut
itself becomes affine in Klein coordinates.

The oracle and cutting-plane perspective is also rooted in the classical theory
of efficient convex programming.  Nemirovski's lecture notes
\cite{Nemirovski1994} present the first-order and separation-oracle viewpoint
in which localization replaces explicit descent as the central algorithmic
primitive.  This viewpoint is especially natural for nonsmooth objectives,
where a subgradient may not define a stable descent step but still defines a
valid separating cut.  In the present work, the Riemannian first-order oracle
returns
\[
        (f(x),g),
        \qquad
        g\in\partial f(x),
\]
and the Lorentz identity
\[
        \ip{g}{\log_XY}_{X}
        =
        \frac{\theta}{\kappaC^2\sinh\theta}\ipL{g}{Y}
\]
turns this oracle response into a Euclidean central cut in one fixed chart.

The geometric and combinatorial foundations of the ellipsoid method are
classically developed by Gr\"otschel, Lov\'asz, and Schrijver
\cite{GLS1993}.  Their separation--optimization paradigm explains why the
ellipsoid update is powerful: once every cut is valid for the minimizer set,
success follows from a dimension-dependent volume decrease and a
size-to-accuracy argument.  Our localization lemma is a variant of this
principle adapted to sublevel-cut oracles rather than Euclidean convex
subgradients.  The central-cut update remains the standard one,
\[
        b=\frac{Qa}{\sqrt{a^{\mathsf T}Qa}},
        \qquad
        c^+=c-\frac1{d+1}b,
        \qquad
        Q^+=\frac{d^2}{d^2-1}
        \left(Q-\frac2{d+1}bb^{\mathsf T}\right),
\]
but the source of the cut is Riemannian and the fixed localization body lives
in the Klein image of the hyperbolic ball.

The geometric convexity assumptions used here are part of the broader theory
of convex analysis in spaces of nonpositive curvature.  Ba\v{c}\'ak
\cite{Bacak2014} develops convex analysis and optimization in Hadamard spaces,
where geodesic convexity, metric projections, and nonsmooth variational
arguments can be formulated without relying on linear structure.  Hyperbolic
space is a smooth Hadamard manifold, so these ideas provide the conceptual
background for Lipschitz geodesically convex objectives.  The present paper
uses a more specialized setting: rather than working intrinsically with general
Hadamard-space tools, it exploits the hyperboloid and Klein models to convert
the relevant intrinsic cuts into affine Euclidean cuts.

The Riemannian optimization notation and first-order framework are aligned with
the modern smooth-manifold optimization literature.  Boumal \cite{Boumal2023}
presents the basic objects used throughout this paper, including tangent
spaces, Riemannian metrics, exponential and logarithm maps, and first-order
optimality notions on manifolds.  Our setting is nonsmooth, but the central
inequality
\[
        f(y)\ge f(x)+\ip{\xi}{\log_x(y)}_x,
        \qquad
        \xi\in\partial f(x),
\]
is the natural geodesic analogue of the Euclidean subgradient inequality.
The main technical point is that, in hyperbolic space, this intrinsic inequality
has an exact projective representation in the Klein chart.

The global nonpositive-curvature background is also connected to the metric
geometry of CAT(0) spaces and Hadamard spaces.  Bridson and Haefliger
\cite{BridsonHaefliger1999} give the standard reference for metric spaces of
nonpositive curvature, including uniqueness of geodesics and convexity
phenomena that underlie Hadamard geometry.  These properties justify the
geometric convexity framework, but they do not by themselves yield a global
Euclidean localization scheme.  The distinction is important: nonpositive
curvature gives a robust intrinsic convexity theory, whereas the one-shot
cutting-plane proof requires the stronger projective fact that hyperbolic
geodesics become Euclidean chords in Klein coordinates.

For the differential-geometric conventions used in the paper, including
Riemannian metrics, curvature, geodesics, and the exponential map, we rely on
the standard framework of do Carmo \cite{doCarmo1992}.  In the hyperboloid
model with metric scaled by $\kappaC^{-2}$, the sectional curvature is
$-\kappaC^2$, and the hyperbolic distance is expressed through the Lorentz
inner product by
\[
        \dist_{\HH}(X,Y)
        =
        \kappaC^{-1}\arcosh(-\ipL{X}{Y}).
\]
These conventions fix the normalization of all constants in the query bound,
especially the dimensionless radius $s=\kappaC r$ and the target scale
$\eps Mr$.

The specific hyperbolic models used here are classical.  Ratcliffe
\cite{Ratcliffe2006} provides a standard reference for the hyperboloid model,
the Beltrami--Klein model, and the projective description of hyperbolic
geodesics.  The paper relies on two consequences of this model.  First, the
closed hyperbolic ball $\bar B_{\HH}(o,r)$ has Klein image
\[
        \bar B_{\R^d}(0,\tanh(\kappaC r)).
\]
Second, hyperbolic geodesic segments are mapped to Euclidean line segments.
These facts are not merely geometric background; they are what allow one fixed
Euclidean ellipsoid to localize the entire feasible region.

Finally, the limitation of the method is related to Beltrami-type rigidity.
Matveev \cite{Matveev2006} discusses the geometric content of Beltrami's
theorem, namely that metrics whose unparametrized geodesics are straight lines
are highly constrained and, under the usual hypotheses, have constant sectional
curvature.  This explains why the proof is genuinely hyperbolic rather than a
generic bounded-curvature argument.  A one-chart method with affine images of
Riemannian subgradient halfspaces would force strong projective structure.
Thus, for arbitrary manifolds with curvature only bounded above and below, one
should not expect the same ordinary ellipsoid proof to apply without a new
localization family for curved cuts.
\section{Hyperboloid model, Klein coordinates, and geodesic chords}

Let
\[
        \ipL{U}{V}=-U_0V_0+\sum_{i=1}^dU_iV_i,
        \qquad
        U=(U_0,\bar U),
        \qquad
        V=(V_0,\bar V).
\]
We represent curvature $-\kappaC^2$ hyperbolic space by
\[
        \HH^d_{-\kappaC^2}
        =\{X\in\R^{d+1}:\ipL{X}{X}=-1,
        \ X_0>0\},
        \qquad
        T_X\HH^d_{-\kappaC^2}=\{U:\ipL{U}{X}=0\},
\]
with scaled metric
\[
        \ip{U}{V}_X=\kappaC^{-2}\ipL{U}{V}.
\]
If
\[
        \theta(X,Y)=\arcosh(-\ipL{X}{Y}),
\]
then
\[
        \dist_{\HH}(X,Y)=\kappaC^{-1}\theta(X,Y).
\]
By applying an isometry we assume first that the ball center is
\[
        o=(1,0,\ldots,0).
\]
The Beltrami--Klein map is
\[
        \Phi(X)=\frac{\bar X}{X_0},
        \qquad
        \Phi^{-1}(u)=X(u)=\frac{(1,u)}{\sqrt{1-\norm{u}^2}},
        \qquad
        \B^d=\{u:\norm{u}<1\}.
\]
For $u\in\B^d$,
\[
        -\ipL{o}{X(u)}=\frac1{\sqrt{1-\norm{u}^2}},
\]
and hence
\[
        \dist_{\HH}(o,X(u))
        =\kappaC^{-1}\arcosh\left(\frac1{\sqrt{1-\norm{u}^2}}\right)
        =\kappaC^{-1}\artanh\norm{u}.
\]
Therefore
\[
        \Phi\bigl(\bar B_{\HH}(o,r)\bigr)
        =\bar B_{\R^d}(0,R_s),
        \qquad
        R_s=\tanh s,
        \qquad
        s=\kappaC r.
\]

\begin{lemma}[geodesics are Klein chords]\label{lem:geodesics-chords}
For $u,v\in\B^d$, the hyperbolic geodesic segment from $X(u)$ to $X(v)$ has Klein image
\[
        [u,v]=\{(1-t)u+tv:0\le t\le1\}.
\]
Consequently, if $C\subset\HH^d_{-\kappaC^2}$ is geodesically convex, then $\Phi(C)$ is Euclidean convex.
\end{lemma}

\begin{proof}
Let
\[
        P=\operatorname{span}\{X(u),X(v)\}\subset\R^{d+1}.
\]
The geodesic is the connected component of
\[
        \HH^d_{-\kappaC^2}\cap P.
\]
Central projection from the origin to the affine plane $X_0=1$ sends every point $Z=(Z_0,\bar Z)$ with $Z_0>0$ to
\[
        \frac{Z}{Z_0}=(1,\Phi(Z)).
\]
Since $P\cap\{X_0=1\}$ is the affine line through $(1,u)$ and $(1,v)$, the projected conic arc is exactly the Euclidean segment joining $u$ and $v$.  If $C$ contains every hyperbolic segment joining two of its points, then $\Phi(C)$ contains every Euclidean segment joining the corresponding Klein points.
\end{proof}

\begin{remark}[why the objective is not pulled back to a convex function]\label{rem:not-convex}
For
\[
        F(u)=f(X(u)),
\]
sublevel sets satisfy
\[
        \{u:F(u)\le\alpha\}
        =\Phi\bigl(\{X:f(X)\le\alpha\}\bigr),
\]
so $F$ is Euclidean quasiconvex.  In general $F$ need not be Euclidean convex.  If
\[
        u_t=(1-t)u+tv,
\]
then $X(u_t)$ lies on the geodesic from $X(u)$ to $X(v)$, but $t$ is not usually proportional to arclength.  Geodesic convexity gives
\[
        f(\gamma(\lambda))\le(1-\lambda)f(\gamma(0))+\lambda f(\gamma(1))
\]
for arclength fraction $\lambda$, not necessarily for the affine parameter $t$ in the Klein chord.  The proof below therefore uses valid sublevel cuts rather than Euclidean convexity of $F$.
\end{remark}

\section{The Lorentz identity and exact Euclidean central cuts}

The key computation is that Riemannian subgradient halfspaces are affine in Klein coordinates.

\begin{proposition}[Lorentz linearization of Riemannian cuts]\label{prop:lorentz-cut}
Let
\[
        X,Y\in\HH^d_{-\kappaC^2},
        \qquad
        g\in T_X\HH^d_{-\kappaC^2},
        \qquad
        \theta=\arcosh(-\ipL{X}{Y}).
\]
Then
\[
        \log_X(Y)=\frac{\theta}{\sinh\theta}\bigl(Y-\cosh\theta\,X\bigr),
\]
and
\[
        \ip{g}{\log_X(Y)}_X
        =\frac{\theta}{\kappaC^2\sinh\theta}\ipL{g}{Y}.
\]
If
\[
        c=\Phi(X),
        \qquad
        u=\Phi(Y),
        \qquad
        \gbar=(g_1,\ldots,g_d)^{\mathsf T},
\]
then
\[
        \ip{g}{\log_X(Y)}_X\le0
        \quad\Longleftrightarrow\quad
        \gbar^{\mathsf T}(u-c)\le0.
\]
Thus the Klein image of the halfspace $\{Y:\ip{g}{\log_X(Y)}_X\le0\}$ is the Euclidean central halfspace
\[
        \{u\in\B^d:\gbar^{\mathsf T}(u-c)\le0\}.
\]
\end{proposition}

\begin{proof}
First,
\[
        \ipL{X}{Y-\cosh\theta X}
        =\ipL{X}{Y}-\cosh\theta\ipL{X}{X}
        =-\cosh\theta+\cosh\theta=0,
\]
so
\[
        Y-\cosh\theta X\in T_X\HH^d_{-\kappaC^2}.
\]
Its Lorentz norm is
\begin{align*}
        \ipL{Y-\cosh\theta X}{Y-\cosh\theta X}
        &=-1-2\cosh\theta\ipL{X}{Y}+\cosh^2\theta\ipL{X}{X}       \\
        &=-1+2\cosh^2\theta-\cosh^2\theta                         \\
        &=\sinh^2\theta.
\end{align*}
Consequently
\[
        V=\frac{\theta}{\sinh\theta}(Y-\cosh\theta X)
\]
has
\[
        \norm{V}_X
        =\kappaC^{-1}\theta
        =\dist_{\HH}(X,Y),
\]
and it points along the unique geodesic from $X$ to $Y$; hence $V=\log_X(Y)$.  Since $g$ is tangent at $X$,
\[
        \ipL{g}{X}=0,
\]
and therefore
\begin{align*}
        \ip{g}{\log_X(Y)}_X
        &=\kappaC^{-2}\frac{\theta}{\sinh\theta}
          \ipL{g}{Y-\cosh\theta X}                         \\
        &=\frac{\theta}{\kappaC^2\sinh\theta}\ipL{g}{Y}.
\end{align*}
The coefficient is positive for $Y\ne X$, and the case $Y=X$ is obtained by continuity.

Now write
\[
        X=X(c)=\frac{(1,c)}{\sqrt{1-\norm{c}^2}},
        \qquad
        Y=X(u)=\frac{(1,u)}{\sqrt{1-\norm{u}^2}}.
\]
Then
\[
        \ipL{g}{Y}
        =\frac{-g_0+\gbar^{\mathsf T}u}{\sqrt{1-\norm{u}^2}}.
\]
Tangency at $X=X(c)$ gives
\[
        0=\ipL{g}{X(c)}
        =\frac{-g_0+\gbar^{\mathsf T}c}{\sqrt{1-\norm{c}^2}},
\]
so
\[
        -g_0+\gbar^{\mathsf T}u=\gbar^{\mathsf T}(u-c).
\]
Because
\[
        \sqrt{1-\norm{u}^2}>0,
        \qquad
        \frac{\theta}{\kappaC^2\sinh\theta}>0,
\]
the sign of $\ip{g}{\log_X(Y)}_X$ is the sign of $\gbar^{\mathsf T}(u-c)$.
\end{proof}

\begin{corollary}[sublevel cuts]\label{cor:sublevel-cut}
Let $D\subseteq\HH^d_{-\kappaC^2}$ be geodesically convex, let $f:D\to\R$ be geodesically convex, and let $g\in\partial f(X)$ at $X\in D$.  If $c=\Phi(X)$, then
\[
        \Phi\bigl(\{Y\in D:f(Y)\le f(X)\}\bigr)
        \subseteq
        \{u:\gbar^{\mathsf T}(u-c)\le0\}.
\]
In particular every minimizer $X^*$ of $f$ over $D$ satisfies
\[
        \gbar^{\mathsf T}(\Phi(X^*)-c)\le0.
\]
\end{corollary}

\begin{proof}
For every $Y\in D$ with $f(Y)\le f(X)$, the subgradient inequality gives
\[
        \ip{g}{\log_X(Y)}_X
        \le f(Y)-f(X)
        \le0.
\]
The result follows from \Cref{prop:lorentz-cut}.
\end{proof}

\begin{remark}[zero subgradient]\label{rem:zero-subgradient}
If $\gbar=0$, then tangency gives
\[
        g_0=0,
        \qquad
        g=0.
\]
If $0\in\partial f(X)$, then
\[
        f(Y)\ge f(X)+\ip{0}{\log_X(Y)}_X=f(X),
        \qquad
        Y\in D,
\]
so $X$ is globally optimal over $D$.
\end{remark}

\section{Global metric distortion in the Klein chart}

The Euclidean localization proof also needs a Lipschitz constant for the pulled-back objective.  It is here, and only here, that the large radius enters.

\begin{lemma}[Klein differential and metric tensor]\label{lem:klein-differential}
For $u\in\B^d$ and $h\in\R^d$,
\[
        DX(u)[h]
        =\frac{(0,h)}{\sqrt{1-\norm{u}^2}}
        +\frac{(1,u)\,u^{\mathsf T}h}{(1-\norm{u}^2)^{3/2}}.
\]
Moreover
\[
        \norm{DX(u)[h]}_{X(u)}^2
        =\kappaC^{-2}\left(
          \frac{\norm{h}^2}{1-\norm{u}^2}
          +\frac{(u^{\mathsf T}h)^2}{(1-\norm{u}^2)^2}
          \right).
\]
Hence
\[
        \norm{DX(u)}_{\ell_2\to\HH}
        =\frac1{\kappaC(1-\norm{u}^2)}.
\]
\end{lemma}

\begin{proof}
Let
\[
        \alpha(u)=(1-\norm{u}^2)^{-1/2},
        \qquad
        X(u)=\alpha(u)(1,u).
\]
Then
\[
        D\alpha(u)[h]
        =(1-\norm{u}^2)^{-3/2}u^{\mathsf T}h,
\]
and therefore
\[
        DX(u)[h]=D\alpha(u)[h](1,u)+\alpha(u)(0,h),
\]
which is the displayed formula.

Set
\[
        q=u^{\mathsf T}h,
        \qquad
        \rho=\norm{u}^2,
        \qquad
        \alpha=(1-\rho)^{-1/2}.
\]
Then
\[
        DX(u)[h]=(\alpha^3q,\alpha h+\alpha^3q u).
\]
The Lorentz norm is
\begin{align*}
        \ipL{DX[h]}{DX[h]}
        &=-(\alpha^3q)^2+\norm{\alpha h+\alpha^3q u}^2        \\
        &=-\alpha^6q^2+\alpha^2\norm{h}^2
          +2\alpha^4q^2+\alpha^6q^2\rho                   \\
        &=\alpha^2\norm{h}^2+\bigl(-\alpha^6(1-\rho)+2\alpha^4\bigr)q^2 \\
        &=\alpha^2\norm{h}^2+\bigl(-\alpha^4+2\alpha^4\bigr)q^2 \\
        &=\frac{\norm{h}^2}{1-\norm{u}^2}
          +\frac{(u^{\mathsf T}h)^2}{(1-\norm{u}^2)^2}.
\end{align*}
Multiplication by $\kappaC^{-2}$ gives the metric expression.  The associated Euclidean matrix is
\[
        G(u)=\kappaC^{-2}\left(
        \frac1{1-\norm{u}^2}\Id
        +\frac1{(1-\norm{u}^2)^2}uu^{\mathsf T}
        \right).
\]
Its eigenvalues are
\[
        \lambda_\perp(u)=\kappaC^{-2}\frac1{1-\norm{u}^2}
        \quad(d-1\text{ times}),
        \qquad
        \lambda_\parallel(u)=\kappaC^{-2}\frac1{(1-\norm{u}^2)^2},
\]
so the operator norm is
\[
        \sqrt{\lambda_\parallel(u)}=\frac1{\kappaC(1-\norm{u}^2)}.
\]
\end{proof}

\begin{corollary}[Euclidean Lipschitz constant]\label{cor:euclidean-lipschitz}
If $f$ is $M$-Lipschitz on $\bar B_{\HH}(o,r)$ and
\[
        F(u)=f(X(u)),
        \qquad
        u\in\bar B(0,R_s),
        \qquad
        R_s=\tanh s,
\]
then $F$ is $L_s$-Lipschitz with
\[
        L_s=\frac{M}{\kappaC(1-R_s^2)}
        =\frac{M\cosh^2s}{\kappaC}.
\]
\end{corollary}

\begin{proof}
For $u,v\in\bar B(0,R_s)$, set
\[
        \ell(t)=u+t(v-u),
        \qquad
        0\le t\le1.
\]
Convexity of $\bar B(0,R_s)$ gives $\norm{\ell(t)}\le R_s$.  Thus
\begin{align*}
        \dist_{\HH}(X(u),X(v))
        &\le\int_0^1\norm{DX(\ell(t))[v-u]}_{X(\ell(t))}\,dt       \\
        &\le\int_0^1\frac{\norm{v-u}}{\kappaC(1-R_s^2)}\,dt          \\
        &=\frac{\norm{v-u}}{\kappaC(1-R_s^2)}.
\end{align*}
Therefore
\[
        \abs{F(u)-F(v)}
        \le M\dist_{\HH}(X(u),X(v))
        \le L_s\norm{u-v}.
\]
Since
\[
        1-R_s^2=1-\tanh^2s=\sech^2s,
\]
we have
\[
        \frac1{1-R_s^2}=\cosh^2s.
\]
\end{proof}

\section{Euclidean localization with sublevel cuts}

The Klein pullback is generally quasiconvex, so the Euclidean lemma must be formulated for sublevel cuts rather than subgradients of a convex function.

\begin{definition}[sublevel-cut oracle]\label{def:sublevel-cut-oracle}
Let
\[
        \calB_R=\bar B(0,R)\subset\R^d,
        \qquad
        F:\calB_R\to\R.
\]
A sublevel-cut oracle returns, at a query $c\in\calB_R$, either a certificate that $c$ is a minimizer or a vector $a(c)\ne0$ such that
\[
        \{u\in\calB_R:F(u)\le F(c)\}
        \subseteq
        \{u:a(c)^{\mathsf T}(u-c)\le0\}.
\]
\end{definition}

\subsection{The one-dimensional branch}

\begin{lemma}[interval localization for $d=1$]\label{lem:interval-sublevel}
Let
\[
        F:[-R,R]\to\R
\]
be $L$-Lipschitz and attain its minimum.  Assume access to a one-dimensional sublevel-cut oracle.  Let
\[
        A=\frac{LR}{\eta},
        \qquad
        A\ge1,
        \qquad
        0<\eta\le LR.
\]
The midpoint interval method returns a queried point $\hat c$ with
\[
        F(\hat c)-F^*\le\eta
\]
after at most
\[
        1+\left\lceil \log_2 A\right\rceil
\]
queries.  In particular this is no more than
\[
        \left\lceil 4\log(16A)\right\rceil.
\]
\end{lemma}

\begin{proof}
Initialize
\[
        I_0=[-R,R].
\]
At step $k\ge0$, query the midpoint
\[
        c_k=\frac{\ell_k+r_k}{2},
        \qquad
        I_k=[\ell_k,r_k].
\]
If the oracle returns $a_k>0$, then every minimizer $u^*$ satisfies
\[
        a_k(u^*-c_k)\le0
        \quad\Longrightarrow\quad
        u^*\le c_k,
\]
so set
\[
        I_{k+1}=[\ell_k,c_k].
\]
If $a_k<0$, then $u^*\ge c_k$ and set
\[
        I_{k+1}=[c_k,r_k].
\]
Thus
\[
        u^*\in I_k,
        \qquad
        \operatorname{length}(I_k)=2R\,2^{-k},
        \qquad
        \abs{c_k-u^*}\le R\,2^{-k}.
\]
By Lipschitz continuity,
\[
        F(c_k)-F^*
        \le L\abs{c_k-u^*}
        \le LR\,2^{-k}.
\]
If
\[
        k\ge \log_2\frac{LR}{\eta},
\]
then $F(c_k)-F^*\le\eta$.  The number of queries is $k+1$, hence the first claim.

For $A\ge1$,
\[
        4\log(16A)
        \ge 4\log16
        >1,
\]
and
\[
        4\log(16A)-\log_2 A
        =4\log16+\left(4-\frac1{\log2}\right)\log A
        >0,
\]
because $4-1/\log2>0$.  Thus
\[
        1+\lceil\log_2A\rceil
        \le \lceil4\log(16A)\rceil .
\]
\end{proof}

\subsection{The ellipsoid branch for \texorpdfstring{$d\ge2$}{d>=2}}

\begin{lemma}[ellipsoid localization with sublevel cuts]\label{lem:ellipsoid-sublevel}
Let
\[
        d\ge2,
        \qquad
        \calB_R=\bar B(0,R)\subset\R^d,
        \qquad
        F:\calB_R\to\R.
\]
Assume that $F$ is $L$-Lipschitz, attains $F^*$ on $\calB_R$, and has a sublevel-cut oracle.  Suppose
\[
        0<\eta\le4LR.
\]
Run the central-cut ellipsoid method with
\[
        E_0=\calB_R=E(0,R^2\Id),
        \qquad
        E(c,Q)=\{u:(u-c)^{\mathsf T}Q^{-1}(u-c)\le1\}.
\]
If the center $c_k$ is infeasible, $\norm{c_k}>R$, use the feasibility cut
\[
        c_k^{\mathsf T}(u-c_k)\le0.
\]
If $c_k\in\calB_R$, query the sublevel-cut oracle and use its cut.  Let $\hat c_N$ be the best feasible query among the first $N$ updates.  If
\[
        N\ge 2d(d+1)\log\!\left(\frac{16LR}{\eta}\right),
\]
then
\[
        F(\hat c_N)-F^*\le\eta.
\]
Each ellipsoid update costs $O(d^2)$ arithmetic operations.
\end{lemma}

\begin{proof}
For a central cut
\[
        a^{\mathsf T}(u-c)\le0,
        \qquad
        a\ne0,
\]
the standard update is
\begin{align}
        b&=\frac{Qa}{\sqrt{a^{\mathsf T}Qa}},                                      \label{eq:ell-b}\\
        c^+&=c-\frac1{d+1}b,                                                       \label{eq:ell-c}\\
        Q^+&=\frac{d^2}{d^2-1}\left(Q-\frac{2}{d+1}bb^{\mathsf T}\right).           \label{eq:ell-Q}
\end{align}
For $d\ge2$ this is well-defined and satisfies
\[
        E(c,Q)\cap\{u:a^{\mathsf T}(u-c)\le0\}
        \subseteq E(c^+,Q^+),
\]
with
\[
        \vol(E(c^+,Q^+))
        \le \exp\!\left(-\frac1{2(d+1)}\right)\vol(E(c,Q)).
\]
Consequently
\[
        \vol(E_N)
        \le \exp\!\left(-\frac{N}{2(d+1)}\right)\vol(\calB_R).
\]

Assume, for contradiction, that every feasible query satisfies
\[
        F(c_k)>F^*+\eta.
\]
Let
\[
        u^*\in\argmin_{\calB_R}F,
        \qquad
        \rho=\frac{\eta}{8L}.
\]
The hypothesis $\eta\le4LR$ gives
\[
        0<\rho\le\frac R2.
\]
We claim that there is a Euclidean ball $B_\rho$ of radius $\rho$ such that
\[
        B_\rho\subseteq\calB_R,
        \qquad
        F(u)\le F^*+\frac\eta2
        \quad\forall u\in B_\rho.
\]
If
\[
        \dist(u^*,\partial\calB_R)\ge\rho,
\]
take $B_\rho=B(u^*,\rho)$.  Then for $u\in B_\rho$,
\[
        F(u)\le F^*+L\norm{u-u^*}\le F^*+L\rho=F^*+\frac\eta8.
\]
If
\[
        \dist(u^*,\partial\calB_R)<\rho,
\]
set
\[
        \tilde u=\left(1-\frac{2\rho}{R}\right)u^*.
\]
Since $\rho\le R/2$,
\[
        \norm{\tilde u}
        \le \left(1-\frac{2\rho}{R}\right)R
        =R-2\rho,
\]
so
\[
        B(\tilde u,\rho)\subseteq\calB_R.
\]
Moreover
\[
        \norm{\tilde u-u^*}
        =\frac{2\rho}{R}\norm{u^*}
        \le2\rho,
\]
and for $u\in B(\tilde u,\rho)$,
\[
        \norm{u-u^*}
        \le\norm{u-\tilde u}+\norm{\tilde u-u^*}
        \le3\rho.
\]
Thus
\[
        F(u)\le F^*+3L\rho=F^*+\frac{3\eta}{8}<F^*+\frac\eta2.
\]
The claim follows.

For an infeasible center $c_k$, the feasibility cut contains all of $\calB_R$, because
\[
        c_k^{\mathsf T}(u-c_k)
        \le \norm{c_k}\norm{u}-\norm{c_k}^2
        \le \norm{c_k}(R-\norm{c_k})<0.
\]
For a feasible center, the contradiction assumption and the construction of $B_\rho$ imply
\[
        B_\rho\subseteq\{u\in\calB_R:F(u)\le F^*+\eta/2\}
        \subseteq\{u\in\calB_R:F(u)\le F(c_k)\}.
\]
The sublevel cut therefore contains $B_\rho$.  By induction over the ellipsoid containments,
\[
        B_\rho\subseteq E_N.
\]
Hence
\[
        \vol(E_N)
        \ge \vol(B_\rho)
        =\left(\frac{\rho}{R}\right)^d\vol(\calB_R).
\]
Combining the upper and lower bounds gives
\[
        \exp\!\left(-\frac{N}{2(d+1)}\right)
        \ge \left(\frac{\rho}{R}\right)^d
        =\left(\frac{\eta}{8LR}\right)^d.
\]
But the assumed lower bound on $N$ yields
\[
        \exp\!\left(-\frac{N}{2(d+1)}\right)
        \le \left(\frac{\eta}{16LR}\right)^d
        <\left(\frac{\eta}{8LR}\right)^d,
\]
a contradiction.  Therefore at least one feasible query has value at most $F^*+\eta$.  The update \eqref{eq:ell-b}--\eqref{eq:ell-Q} consists of one matrix-vector product, one quadratic form, and one rank-one update, hence costs $O(d^2)$ arithmetic operations.
\end{proof}

\section{The global Klein cutting-plane algorithm}

The algorithm is the Euclidean localization method applied to the Klein image
\[
        \calB_s=\bar B(0,R_s),
        \qquad
        R_s=\tanh s.
\]
The only hyperbolic operation inside the localization loop is the conversion
\[
        c_k\mapsto X(c_k)=\frac{(1,c_k)}{\sqrt{1-\norm{c_k}^2}},
        \qquad
        g_k\mapsto a_k=(g_{k,1},\ldots,g_{k,d})^{\mathsf T}.
\]

\begin{algorithm}[global Klein cutting plane]\label{alg:gkcp}
Input:
\[
        d\ge1,
        \quad
        \kappaC>0,
        \quad
        r>0,
        \quad
        \eps\in(0,1),
        \quad
        \text{first-order oracle for }f\text{ on }\bar B_{\HH}(o,r).
\]
Set
\[
        s=\kappaC r,
        \qquad
        R_s=\tanh s,
        \qquad
        \eta=\eps Mr=\eps M\frac{s}{\kappaC},
\]
and
\[
        N=N(d,s,\eps)
        =\left\lceil2d(d+1)\log\left(\frac{16\sinh s\cosh s}{s\eps}\right)\right\rceil.
\]

If $d=1$, run the interval method of \Cref{lem:interval-sublevel} on $[-R_s,R_s]$ with cuts obtained below.  If $d\ge2$, initialize
\[
        c_0=0,
        \qquad
        Q_0=R_s^2\Id_d,
        \qquad
        E_0=E(c_0,Q_0).
\]
For $k=0,\ldots,N-1$:
\begin{enumerate}[leftmargin=2.4em,itemsep=0.25em]
\item Let $c_k$ be the current center.
\item If $\norm{c_k}>R_s$, use the feasibility normal
\[
        a_k=c_k
\]
and the cut
\[
        a_k^{\mathsf T}(u-c_k)
        =c_k^{\mathsf T}(u-c_k)
        \le0.
\]
\item If $\norm{c_k}\le R_s$, query
\[
        X_k=X(c_k)=\frac{(1,c_k)}{\sqrt{1-\norm{c_k}^2}}.
\]
The oracle returns
\[
        f(X_k),
        \qquad
        g_k\in\partial f(X_k).
\]
If $g_k=0$, return $X_k$.  Otherwise set
\[
        a_k=(g_{k,1},\ldots,g_{k,d})^{\mathsf T}
\]
and use the central cut
\[
        a_k^{\mathsf T}(u-c_k)\le0.
\]
Record $X_k$ and its value.
\item If $d\ge2$, update
\[
        b_k=\frac{Q_ka_k}{\sqrt{a_k^{\mathsf T}Q_ka_k}},
        \qquad
        c_{k+1}=c_k-\frac1{d+1}b_k,
\]
\[
        Q_{k+1}=\frac{d^2}{d^2-1}
        \left(Q_k-\frac{2}{d+1}b_kb_k^{\mathsf T}\right).
\]
If $d=1$, update the interval by the sign of $a_k$:
\[
        a_k>0\Rightarrow I_{k+1}=I_k\cap(-\infty,c_k],
        \qquad
        a_k<0\Rightarrow I_{k+1}=I_k\cap[c_k,\infty).
\]
\end{enumerate}
Return the recorded feasible query of smallest objective value.
\end{algorithm}

\begin{proposition}[cut validity]\label{prop:algorithm-cuts-valid}
Every cut used in \Cref{alg:gkcp} contains the Klein image of the minimizer set.  Whenever $c_k$ is feasible and $X_k=X(c_k)$, the cut also contains the complete lower sublevel set
\[
        \Phi\bigl(\{X\in\bar B_{\HH}(o,r):f(X)\le f(X_k)\}\bigr).
\]
\end{proposition}

\begin{proof}
If $\norm{c_k}>R_s$, then for every $u\in\bar B(0,R_s)$,
\[
        c_k^{\mathsf T}(u-c_k)
        \le \norm{c_k}R_s-\norm{c_k}^2
        <0,
\]
so the feasibility cut contains the whole feasible ball.  If $c_k$ is feasible, \Cref{cor:sublevel-cut} gives
\[
        \Phi\bigl(\{X:f(X)\le f(X_k)\}\bigr)
        \subseteq
        \{u:a_k^{\mathsf T}(u-c_k)\le0\}.
\]
If $a_k=0$, then $g_k=0$ by \Cref{rem:zero-subgradient}; the point is optimal and the algorithm stops.
\end{proof}

\section{Proof of the main theorem}

\begin{proof}[Proof of \Cref{thm:main}]
By a hyperbolic isometry assume $x_0=o$.  Define
\[
        F(u)=f(X(u)),
        \qquad
        u\in\calB_s=\bar B(0,R_s),
        \qquad
        R_s=\tanh s.
\]
Then
\[
        \min_{X\in\bar B_{\HH}(o,r)}f(X)
        =\min_{u\in\calB_s}F(u).
\]
By \Cref{cor:euclidean-lipschitz},
\[
        F\text{ is }L_s\text{-Lipschitz},
        \qquad
        L_s=\frac{M\cosh^2s}{\kappaC}.
\]
The target accuracy is
\[
        \eta=\eps Mr=\eps M\frac{s}{\kappaC}.
\]
The normalized ellipsoid ratio is
\begin{align*}
        \frac{L_sR_s}{\eta}
        &=\frac{(M\cosh^2s/\kappaC)\tanh s}{\eps Ms/\kappaC}       \\
        &=\frac{\cosh^2s\tanh s}{\eps s}                         \\
        &=\frac{\sinh s\cosh s}{\eps s}.
\end{align*}
For $s>0$,
\[
        \sinh s\cosh s>s,
\]
because
\[
        \frac{d}{ds}\bigl(\sinh s\cosh s-s\bigr)
        =\cosh^2s+\sinh^2s-1
        =2\sinh^2s\ge0,
\]
and the difference is zero at $s=0$.  Hence
\[
        \eta<L_sR_s,
\]
so the assumptions
\[
        A=\frac{L_sR_s}{\eta}\ge1,
        \qquad
        \eta\le4L_sR_s
\]
needed in \Cref{lem:interval-sublevel,lem:ellipsoid-sublevel} hold automatically.

For $d=1$, \Cref{lem:interval-sublevel} gives accuracy after at most
\[
        \left\lceil4\log\left(16\frac{L_sR_s}{\eta}\right)\right\rceil
        =\left\lceil2d(d+1)\log\left(16\frac{L_sR_s}{\eta}\right)\right\rceil
\]
queries.  For $d\ge2$, \Cref{lem:ellipsoid-sublevel} gives accuracy after at most
\[
        \left\lceil2d(d+1)\log\left(16\frac{L_sR_s}{\eta}\right)\right\rceil
\]
updates and no more feasible oracle calls.  Substituting the ratio gives
\[
        16\frac{L_sR_s}{\eta}
        =\frac{16\sinh s\cosh s}{s\eps},
\]
which is the stated $N(d,s,\eps)$.

For the simple bound, use
\[
        \sinh s\le s e^s,
        \qquad
        \cosh s\le e^s,
        \qquad
        s>0.
\]
Then
\[
        \log\left(\frac{16\sinh s\cosh s}{s\eps}\right)
        \le
        \log\left(\frac{16e^{2s}}{\eps}\right)
        =2s+\log\frac{16}{\eps}.
\]
Finally,
\[
        s+\log(1/\eps)
        \le (1+s)\log(e/\eps),
\]
and
\[
        1+s\le2\zeta_s,
        \qquad
        \zeta_s=\frac{s}{\tanh s},
\]
because for $0<s\le1$,
\[
        1+s\le2\le2\zeta_s,
\]
while for $s\ge1$,
\[
        1+s\le2s\le2\frac{s}{\tanh s}=2\zeta_s.
\]
Thus
\[
        N(d,s,\eps)=O\bigl(d^2(s+\log(1/\eps))\bigr)
        =O\bigl(d^2\zeta_s\log(e/\eps)\bigr)
\]
in the coarser product-style parameterization.  The arithmetic bounds follow from the interval and ellipsoid updates already analyzed.
\end{proof}

\section{Correct constants and asymptotic regimes}

The formula
\[
        \frac{L_sR_s}{\eta}
        =\frac{\sinh s\cosh s}{s\eps}
\]
is exact.  The small-curvature expansion follows from
\[
        \sinh s=s+\frac{s^3}{6}+\frac{s^5}{120}+O(s^7),
        \qquad
        \cosh s=1+\frac{s^2}{2}+\frac{s^4}{24}+O(s^6),
\]
so
\begin{align*}
        \sinh s\cosh s
        &=\left(s+\frac{s^3}{6}+\frac{s^5}{120}+O(s^7)\right)
          \left(1+\frac{s^2}{2}+\frac{s^4}{24}+O(s^6)\right)  \\
        &=s+\left(\frac12+\frac16\right)s^3
          +\left(\frac1{24}+\frac1{12}+\frac1{120}\right)s^5
          +O(s^7)                                                 \\
        &=s+\frac23s^3+\frac{2}{15}s^5+O(s^7).
\end{align*}
Therefore
\[
        \frac{\sinh s\cosh s}{s}
        =1+\frac23s^2+\frac{2}{15}s^4+O(s^6),
\]
and
\[
        \log\left(\frac{\sinh s\cosh s}{s\eps}\right)
        =\log(1/\eps)+\frac23s^2+O(s^4).
\]
For large $s$,
\[
        \sinh s\cosh s
        =\frac{(e^s-e^{-s})(e^s+e^{-s})}{4}
        =\frac{e^{2s}-e^{-2s}}4
        =\frac{e^{2s}}4(1-e^{-4s}).
\]
Consequently
\[
        \frac{\sinh s\cosh s}{s}
        =\frac{e^{2s}}{4s}(1-e^{-4s}),
\]
and
\[
        \log\left(\frac{\sinh s\cosh s}{s\eps}\right)
        =\log(1/\eps)+2s-\log(4s)+\log(1-e^{-4s}).
\]
Since
\[
        \log(1-z)=-z+O(z^2),
        \qquad
        z=e^{-4s},
\]
the correct large-radius expansion is
\[
        \log\left(\frac{\sinh s\cosh s}{s\eps}\right)
        =\log(1/\eps)+2s-\log(4s)+O(e^{-4s}).
\]
The constant is $4s$, not $2s$, inside the logarithm.

\begin{corollary}[Euclidean limiting case]\label{cor:euclidean-limit}
If $\kappaC\downarrow0$ with $r$ fixed, so that $s=\kappaC r\downarrow0$, then
\[
        \frac{\sinh s\cosh s}{s}\to1,
\]
and the bound becomes
\[
        N(d,s,\eps)\to
        \left\lceil2d(d+1)\log\frac{16}{\eps}\right\rceil.
\]
This is the normalized Euclidean ellipsoid bound for target accuracy $\eps Mr$ on a Euclidean ball of radius $r$.
\end{corollary}

\section{Numerical experiments and visual diagnostics}\label{sec:numerics}

This section is a reproducible diagnostic layer for the preceding theorem.  It does not change the proof: the experiments use the same Klein-coordinate cuts, the same central-cut ellipsoid update, and the same normalized target scale
\[
        \frac{f(\hat x)-f^*}{Mr}\le \eps.
\]
All experiments use curvature parameter $\kappaC=1$ and closed feasible ball radius $r=s$, so $M=1$ and the target is $\eta=\eps s$.

\subsection{Visual check of one exact Klein cut}

\Cref{fig:klein-cut-geometry} illustrates the sign identity in \Cref{prop:lorentz-cut} for the distance objective $f(X)=\dist_{\HH}(X,X_*)$.  The lower sublevel set through the queried point is the hyperbolic ball centered at $X_*$ with radius $f(X_k)$; in the Klein disk its boundary is curved, but the Riemannian subgradient cut is an exact Euclidean halfspace through the query.  The picture is generated from the same formulas used by the numerical oracle.

\begin{figure}[t]
\centering
\includegraphics[width=0.72\linewidth]{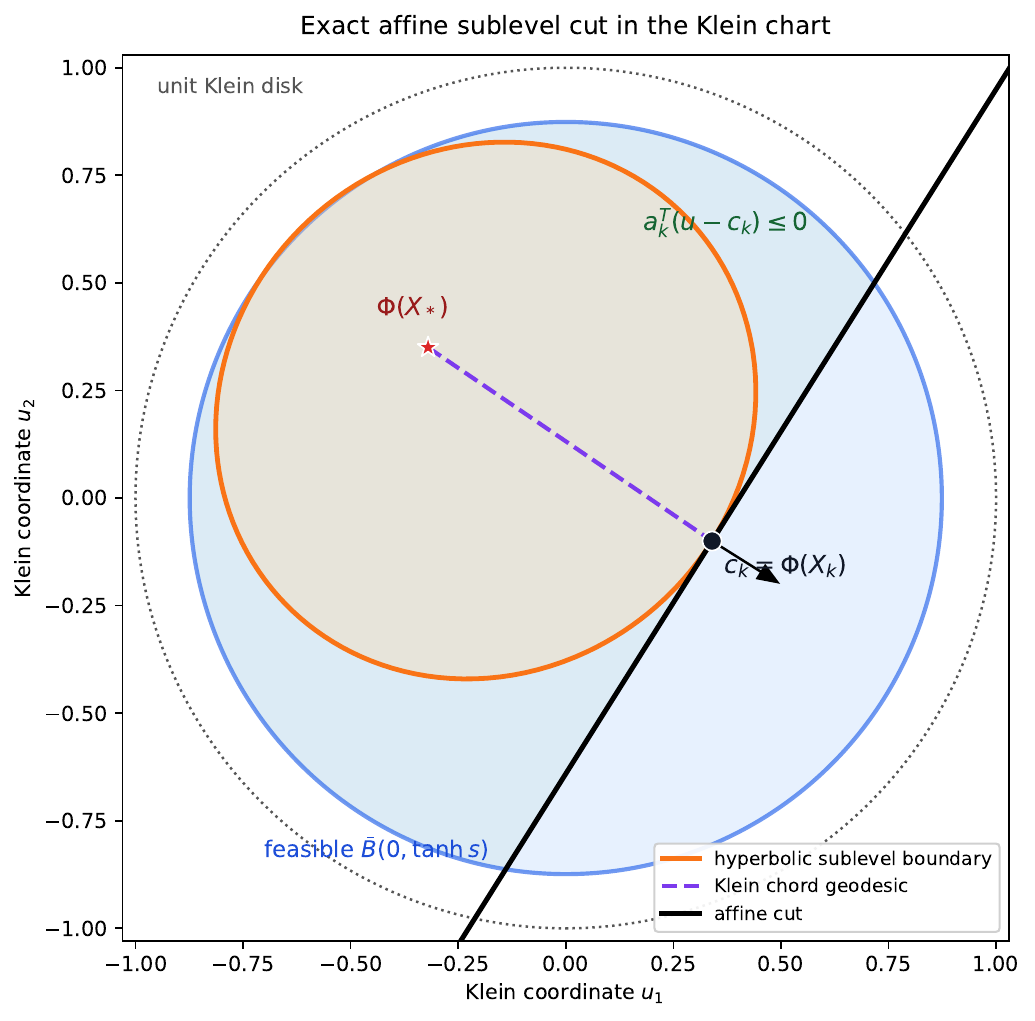}
\caption{A single Klein cutting plane.  The blue disk is the feasible Klein image $\bar B(0,\tanh s)$; the orange curve is a hyperbolic lower sublevel boundary for $f(X)=\dist_{\HH}(X,X_*)$; the black line is the exact Euclidean central cut $a_k^{\mathsf T}(u-c_k)=0$ obtained from the Lorentz spatial component of the Riemannian subgradient.  Although the pulled-back objective is generally only quasiconvex, the sublevel set lies on the correct side of an affine cut.}
\label{fig:klein-cut-geometry}
\end{figure}

\subsection{Complexity landscape}

\Cref{fig:complexity-landscape} visualizes the exact logarithmic factor
\[
        \ell(s,\eps)=\log\!\left(\frac{16\sinh s\cosh s}{s\eps}\right),
        \qquad
        N(d,s,\eps)=\lceil 2d(d+1)\ell(s,\eps)\rceil.
\]
The heat map separates the intrinsic dimension multiplier $2d(d+1)$ from the curvature--accuracy factor.  The line plot compares $\ell(s,\eps)$ with the large-radius expansion
\[
        \log(16/\eps)+2s-\log(4s),
\]
showing that the additive dependence $s+\log(1/\eps)$ is already visible at moderate radius.

\begin{figure}[t]
\centering
\includegraphics[width=0.98\linewidth]{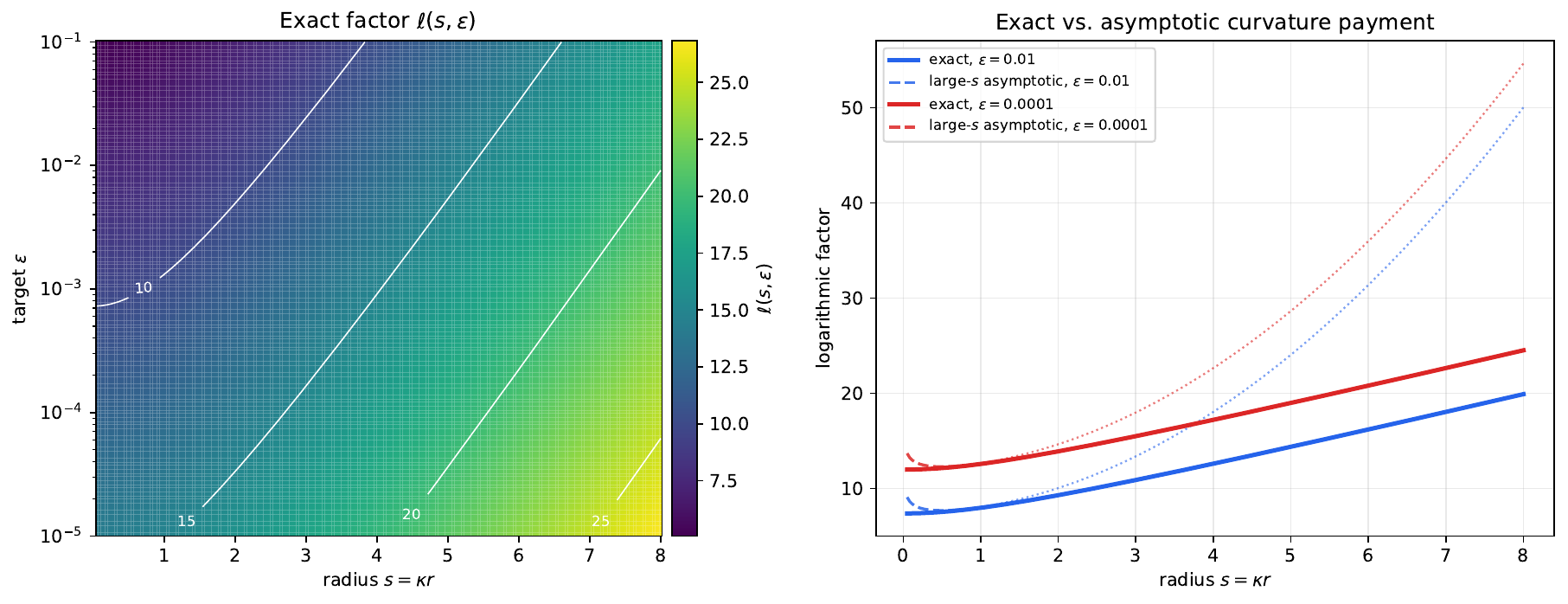}
\caption{Exact complexity factor for the one-shot Klein cutting-plane method.  Left: heat map of $\ell(s,\eps)$ over radius $s=\kappaC r$ and target accuracy $\eps$.  Right: exact curves and the large-radius asymptotic, confirming the additive payment $2s+\log(1/\eps)-\log(4s)$ rather than a chained local product in $s$ and $\log(1/\eps)$.}
\label{fig:complexity-landscape}
\end{figure}

\subsection{Reproducible nonsmooth test family}

For a controlled nonsmooth geodesically convex benchmark, choose a target minimizer $X_*$ inside the feasible ball and a Lorentz-orthonormal tangent frame $\{v_i\}_{i=1}^d\subset T_{X_*}\HH^d$.  For a fixed spread parameter $\tau>0$, define
\[
        Y_i^\pm=\cosh(\tau)X_*\pm\sinh(\tau)v_i,
        \qquad
        f(X)=\max_{1\le i\le d,\,\sigma\in\{+,-\}}
        \dist_{\HH}(X,Y_i^\sigma).
\]
This is the pointwise maximum of $2d$ geodesically convex $1$-Lipschitz distance functions, hence it is geodesically convex and $1$-Lipschitz.  Moreover $f^*=\tau$ at $X_*$: for every pair $Y_i^+,Y_i^-$, the triangle inequality gives
\[
        \max\{\dist(X,Y_i^+),\dist(X,Y_i^-)\}
        \ge \frac12\dist(Y_i^+,Y_i^-)=\tau,
\]
with equality at $X=X_*$.  At each query, the oracle returns the subgradient of an active distance term.

\begin{figure}[t]
\centering
\includegraphics[width=0.82\linewidth]{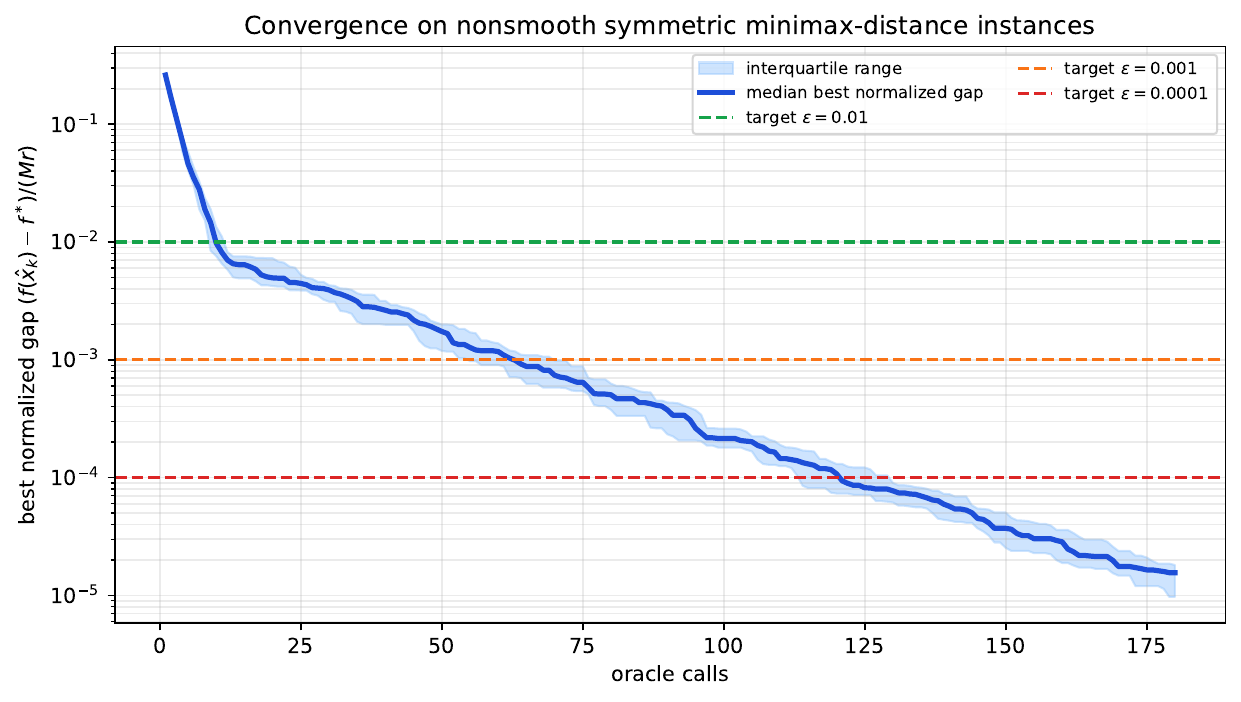}
\caption{Empirical convergence on the nonsmooth symmetric minimax-distance benchmark with $d=4$, $s=2$, $\tau=0.8$, and $30$ random target locations.  The solid curve is the median best normalized gap $(f(\hat x_k)-f^*)/(Mr)$; the band shows the interquartile range.  The dashed horizontal lines are representative target accuracies.}
\label{fig:experiment-convergence}
\end{figure}

\begin{figure}[t]
\centering
\includegraphics[width=0.98\linewidth]{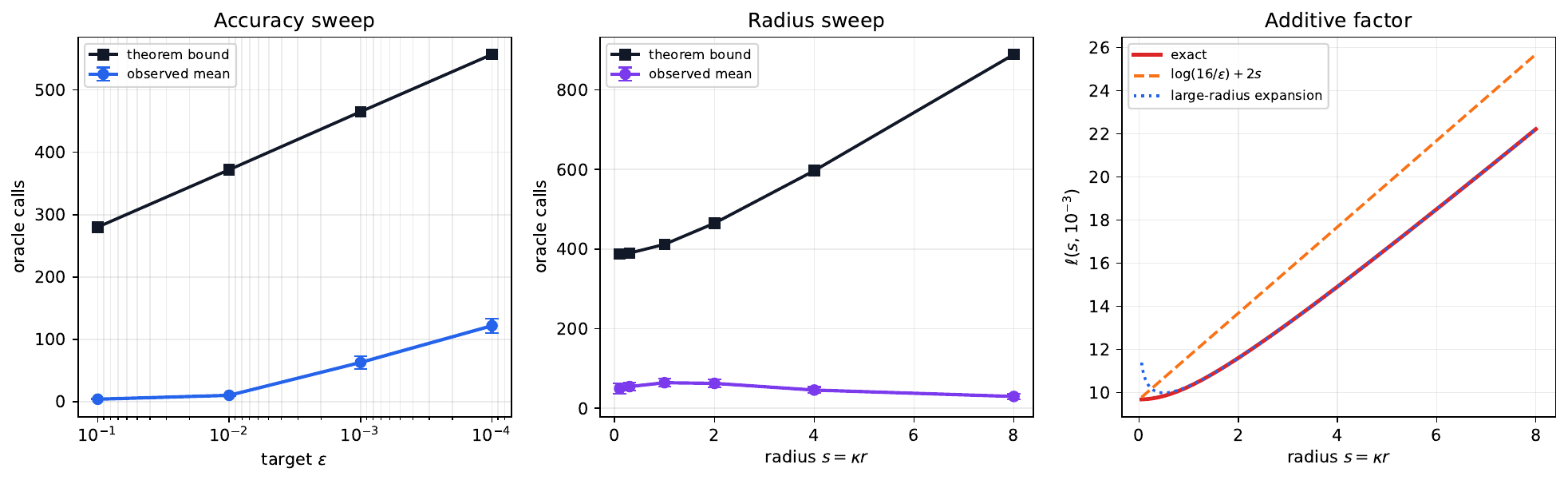}
\caption{Empirical scaling of oracle calls on the nonsmooth benchmark, averaged over $20$ random target locations.  The theorem gives a worst-case envelope; the observed counts stay below that envelope in all displayed runs.  The right panel plots the same additive curvature--accuracy factor that appears in the proof.}
\label{fig:experiment-scaling}
\end{figure}

\begin{table}[t]
\caption{Empirical query counts for the nonsmooth minimax-distance benchmark}
\label{tab:numerical-query-counts}
\centering
\small
\begin{tabular}{llrrrr}
\toprule
Sweep & Value & Mean queries & Std. dev. & Max queries & Theorem $N$ \\
\midrule
$d$ & $2$  & $26.0$  & $5.3$  & $32$  & $140$ \\
$d$ & $4$  & $62.8$  & $10.5$ & $80$  & $465$ \\
$d$ & $8$  & $75.0$  & $34.8$ & $129$ & $1671$ \\
$d$ & $16$ & $56.7$  & $5.9$  & $66$  & $6311$ \\
\midrule
$s$ & $0.1$ & $49.9$ & $13.0$ & $67$ & $388$ \\
$s$ & $0.3$ & $54.8$ & $10.7$ & $75$ & $390$ \\
$s$ & $1$   & $64.5$ & $10.1$ & $79$ & $412$ \\
$s$ & $2$   & $62.8$ & $10.5$ & $80$ & $465$ \\
$s$ & $4$   & $46.0$ & $7.7$  & $56$ & $597$ \\
$s$ & $8$   & $29.9$ & $7.4$  & $42$ & $889$ \\
\midrule
$\eps$ & $10^{-1}$ & $3.9$   & $0.4$  & $4$   & $280$ \\
$\eps$ & $10^{-2}$ & $10.1$  & $1.7$  & $13$  & $372$ \\
$\eps$ & $10^{-3}$ & $62.8$  & $10.5$ & $80$  & $465$ \\
$\eps$ & $10^{-4}$ & $121.7$ & $11.6$ & $141$ & $557$ \\
\bottomrule
\end{tabular}
\end{table}

The table uses $\tau=0.8$ and target location norm $\|\Phi(X_*)\|=0.55\tanh s$ unless the swept variable changes the radius.  For the dimension sweep, $s=2$ and $\eps=10^{-3}$; for the radius sweep, $d=4$ and $\eps=10^{-3}$; for the accuracy sweep, $d=4$ and $s=2$.  These instances are not meant to estimate worst-case constants.  They check three implementation-sensitive points: exact affine cuts in a fixed chart, stable $O(d^2)$ ellipsoid updates, and the predicted additive curvature--accuracy factor in the theorem.

\FloatBarrier
\section{Relation to the chained local construction}

A local construction uses charts of radius comparable to $1/\kappaC$.  In schematic form one takes
\[
        R_{\rm loc}=\frac1{\kappaC},
\]
solves local subproblems of the form
\[
        \min\{f(x):x\in\bar B_{\HH}(x_k,R_{\rm loc})\cap\bar B_{\HH}(x_0,r)\},
\]
and obtains progress resembling
\[
        f(x_{k+1})-f^*
        \le
        \left(1-c\frac{R_{\rm loc}}{r}\right)(f(x_k)-f^*)
        +O(\eps MR_{\rm loc}).
\]
Thus the number of local phases is
\[
        T=O\left(\frac r{R_{\rm loc}}\log\frac1\eps\right)
        =O\left(s\log\frac1\eps\right).
\]
Each phase costs
\[
        O\left(d^2\log\frac1\eps\right),
\]
so the chained dependence is
\[
        O\left(sd^2\log^2\frac1\eps\right)
\]
up to the usual replacement of $s$ by $\zeta_s=\Theta(1+s)$.

The one-shot method pays instead
\[
        \log\left(\frac{L_sR_s}{\eta}\right)
        =\log\left(\frac{\sinh s\cosh s}{s\eps}\right)
        =\log(1/\eps)+2s-\log(4s)+O(e^{-4s}).
\]
Thus the interaction between curvature radius and accuracy changes from
\[
        s\log^2(1/\eps)
\]
to
\[
        s+\log(1/\eps).
\]
Thus, within the hyperbolic model-space setting, the sharp additive estimate implies the coarser product-style estimate
\[
        d^2(s+\log(1/\eps))
        \le
        d^2(1+s)\log(e/\eps)
        \le
        2d^2\zeta_s\log(e/\eps).
\]
The word ``equivalently'' should not be used for the last inequality: it is a coarser corollary within the hyperbolic setting, not an equivalent expression for the sharp additive bound.

\section{Base-point-free implementation}

For an arbitrary center $x_0\in\HH^d_{-\kappaC^2}$, choose a Lorentz-orthonormal frame
\[
        E_0=x_0,
        \qquad
        \ipL{E_0}{E_0}=-1,
        \qquad
        \ipL{E_0}{E_i}=0,
        \qquad
        \ipL{E_i}{E_j}=\delta_{ij}.
\]
This is a Lorentz-orthonormal spatial frame, not a Riemannian-orthonormal one.  If $e_i$ is Riemannian-orthonormal, then
\[
        \ip{e_i}{e_j}_{x_0}=\delta_{ij}
        \quad\Longleftrightarrow\quad
        \ipL{e_i}{e_j}=\kappaC^2\delta_{ij},
\]
and the Lorentz-normalized vector is
\[
        E_i=\kappaC^{-1}e_i.
\]

The base-point-free Klein coordinate is
\[
        \Phi_{x_0}(X)
        =\frac{(\ipL{X}{E_1},\ldots,\ipL{X}{E_d})^{\mathsf T}}{-\ipL{X}{E_0}},
\]
with inverse
\[
        X_{x_0}(u)
        =\frac{E_0+\sum_{i=1}^du_iE_i}{\sqrt{1-\norm{u}^2}}.
\]
Then
\[
        \Phi_{x_0}(\bar B_{\HH}(x_0,r))=\bar B(0,\tanh(\kappaC r)).
\]
If the oracle returns $g\in T_X\HH^d_{-\kappaC^2}$ at
\[
        X=X_{x_0}(c),
\]
the Euclidean cut normal is
\[
        a(g)=\bigl(\ipL{g}{E_1},\ldots,\ipL{g}{E_d}\bigr)^{\mathsf T}.
\]
Indeed, for
\[
        Y=X_{x_0}(u),
\]
we have
\[
        \ipL{g}{Y}
        =\frac{\ipL{g}{E_0}+\sum_{i=1}^du_i\ipL{g}{E_i}}{\sqrt{1-\norm{u}^2}},
\]
while tangency at $X_{x_0}(c)$ gives
\[
        0=\ipL{g}{X_{x_0}(c)}
        =\frac{\ipL{g}{E_0}+\sum_{i=1}^dc_i\ipL{g}{E_i}}{\sqrt{1-\norm{c}^2}}.
\]
Subtracting yields
\[
        \ipL{g}{Y}
        =\frac{a(g)^{\mathsf T}(u-c)}{\sqrt{1-\norm{u}^2}},
\]
and hence
\[
        \ip{g}{\log_XY}_{X}\le0
        \quad\Longleftrightarrow\quad
        a(g)^{\mathsf T}(u-c)\le0.
\]

\section{Why the one-chart proof is hyperbolic}

The proof works because the Klein chart has two exact properties:
\[
        \text{geodesics}\mapsto\text{affine lines},
        \qquad
        \text{Riemannian subgradient halfspaces}\mapsto\text{affine halfspaces}.
\]
A chart with the second property is highly restrictive.

\begin{definition}[affine-cut chart]\label{def:affine-cut-chart}
Let $U\subset\calM$ be geodesically convex.  A smooth chart $\psi:U\to\Omega\subset\R^d$ is an affine-cut chart if, for every $x\in U$ and every tangent vector $g\in T_x\calM$, there are $a_{x,g}\in\R^d$ and $\beta_{x,g}\in\R$ such that
\[
        \{y\in U:\ip{g}{\log_x(y)}_x=0\}
        =\{y\in U:a_{x,g}^{\mathsf T}\psi(y)=\beta_{x,g}\}.
\]
\end{definition}

\begin{proposition}[affine cuts force projective geodesics]\label{prop:affine-cuts-projective}
Assume $\psi$ is an affine-cut chart and the displayed equality holds locally for all tangent vectors.  Then every sufficiently short geodesic segment in $U$ has image contained in a Euclidean line.
\end{proposition}

\begin{proof}
Fix
\[
        x\in U,
        \qquad
        v\in T_x\calM,
        \qquad
        v\ne0,
\]
and let
\[
        \gamma(t)=\exp_x(tv)
\]
for $|t|$ small.  For every $g\in T_x\calM$ with
\[
        \ip{g}{v}_x=0,
\]
we have the exact identity
\[
        \log_x(\gamma(t))=tv,
\]
so
\[
        \ip{g}{\log_x(\gamma(t))}_x=t\ip{g}{v}_x=0.
\]
Thus $\psi(\gamma(t))$ lies in every hyperplane
\[
        a_{x,g}^{\mathsf T}z=\beta_{x,g},
        \qquad
        g\in v^\perp.
\]
At $t=0$, differentiating the equality of hypersurfaces gives
\[
        a_{x,g}^{\mathsf T}D\psi_x(w)=0
        \quad\Longleftrightarrow\quad
        \ip{g}{w}_x=0.
\]
As $g$ ranges over $v^\perp$, the common tangent intersection is
\[
        \{\lambda D\psi_xv:\lambda\in\R\}.
\]
The common intersection of the corresponding affine hyperplanes is therefore a Euclidean affine line through $\psi(x)$ with direction $D\psi_xv$, and $\psi(\gamma(t))$ is contained in that line for $|t|$ small.
\end{proof}

Beltrami-type rigidity says that a Riemannian metric whose unparametrized geodesics are straight lines in local coordinates has constant sectional curvature under the usual nondegeneracy hypotheses.  This explains why ordinary ellipsoids suffice for the hyperbolic proof but do not automatically extend to arbitrary bounded-curvature manifolds.  A general solution would require a different localization family,
\[
        \calK(\theta)\subset\R^d,
\]
and a computable update for curved cuts,
\[
        \calK(\theta)\cap\{q_a(u)\le0\}\subseteq\calK(\theta^+),
        \qquad
        S(\theta^+)\le\rho S(\theta),
        \qquad
        \rho<1,
\]
with a size-to-accuracy theorem replacing the Euclidean volume argument.

\section{Conclusion}

In the constant negative-curvature setting motivated by the COLT open problem, hyperbolic space admits a one-shot deterministic first-order cutting-plane method.  In exact terms,
\[
        f(\hat x)-f^*
        \le\eps Mr,
\]
after
\[
        N(d,s,\eps)
        =\left\lceil2d(d+1)
        \log\left(\frac{16\sinh s\cosh s}{s\eps}\right)\right\rceil
\]
queries.  The core mechanism is
\[
        \ip{g}{\log_XY}_{X}
        =\frac{\theta}{\kappaC^2\sinh\theta}\ipL{g}{Y},
        \qquad
        \ipL{g}{Y}\text{ has the sign of }a_g^{\mathsf T}(\Phi(Y)-\Phi(X)).
\]
The exact distortion payment is
\[
        L_s=\frac{M\cosh^2s}{\kappaC},
        \qquad
        R_s=\tanh s,
        \qquad
        \eta=\eps M\frac{s}{\kappaC},
\]
so
\[
        \log\frac{L_sR_s}{\eta}
        =\log\left(\frac{\sinh s\cosh s}{s\eps}\right)
        =\log(1/\eps)+2s-\log(4s)+O(e^{-4s}).
\]
Thus the final complexity is
\[
        O\bigl(d^2(s+\log(1/\eps))\bigr),
\]
and, in the coarser product-style parameterization using
\[
        \zeta_s=\frac{s}{\tanh s},
\]
this implies
\[
        O\bigl(d^2\zeta_s\log(e/\eps)\bigr).
\]
The $d=1$ singularity of the central ellipsoid formula is removed by an interval branch; for $d\ge2$ the ellipsoid update is exactly the standard central-cut update.

We emphasize that this result does not settle the general Hadamard-manifold open problem.  The method relies on the Beltrami--Klein projective model of hyperbolic space, where geodesics become Euclidean chords and Riemannian subgradient halfspaces become affine Euclidean halfspaces.  These features are not available, for example, on the manifold of positive definite matrices with its usual nonpositively curved geometries.  Extending deterministic first-order cutting-plane localization to such spaces remains open.
\section*{Declaration of Generative AI and AI-Assisted Technologies in the Writing Process}
During the preparation of this work, the authors used DeepSeek for logical structuring and language polishing. After using this tool, the authors reviewed and edited the content as needed and take full responsibility for the content of the published article.
\appendix

\section{Derivation of the central-cut ellipsoid update}\label{app:ellipsoid-update}

Let
\[
        E(c,Q)=\{u:(u-c)^{\mathsf T}Q^{-1}(u-c)\le1\},
        \qquad
        Q\succ0,
        \qquad
        d\ge2.
\]
For the central cut
\[
        a^{\mathsf T}(u-c)\le0,
        \qquad
        a\ne0,
\]
apply the affine transformation
\[
        z=Q^{-1/2}(u-c).
\]
Then
\[
        E(c,Q)\mapsto\{z:\norm z\le1\},
\]
and the cut becomes
\[
        \tilde a^{\mathsf T}z\le0,
        \qquad
        \tilde a=\frac{Q^{1/2}a}{\sqrt{a^{\mathsf T}Qa}},
        \qquad
        \norm{\tilde a}=1.
\]
After rotation assume $\tilde a=e_1$.  The containing ellipsoid for the half-ball
\[
        \{z:\norm z\le1,\ e_1^{\mathsf T}z\le0\}
\]
has center and shape
\[
        \tilde c=-\frac1{d+1}e_1,
        \qquad
        \tilde Q=\frac{d^2}{d^2-1}\left(\Id-\frac{2}{d+1}e_1e_1^{\mathsf T}\right).
\]
Indeed, the eigenvalue of $\tilde Q$ in the $e_1$ direction is
\[
        \frac{d^2}{(d+1)^2},
\]
and its eigenvalue in each orthogonal direction is
\[
        \frac{d^2}{d^2-1}.
\]
Hence
\[
        \tilde Q^{-1}
        =
        \frac{(d+1)^2}{d^2}e_1e_1^{\mathsf T}
        +\frac{d^2-1}{d^2}\left(\Id-e_1e_1^{\mathsf T}\right).
\]
For $z_1\le0$ and $\norm z\le1$, we have $-1\le z_1\le0$ and
\begin{align*}
        (z-\tilde c)^{\mathsf T}\tilde Q^{-1}(z-\tilde c)
        &=\frac{(d+1)^2}{d^2}\left(z_1+\frac1{d+1}\right)^2
          +\frac{d^2-1}{d^2}\sum_{i=2}^dz_i^2                         \\
        &\le\frac{(d+1)^2}{d^2}\left(z_1+\frac1{d+1}\right)^2
          +\frac{d^2-1}{d^2}(1-z_1^2)                                  \\
        &=1+\frac{2(d+1)}{d^2}(z_1+z_1^2)                              \\
        &=1+\frac{2(d+1)}{d^2}z_1(1+z_1)                               \\
        &\le1.
\end{align*}
Thus
\[
        \{z:\norm z\le1,\ e_1^{\mathsf T}z\le0\}
        \subseteq
        \{z:(z-\tilde c)^{\mathsf T}\tilde Q^{-1}(z-\tilde c)\le1\}.
\]

The volume ratio is
\begin{align*}
        \frac{\vol(E^+)}{\vol(E)}
        &=\det(\tilde Q)^{1/2}                                                \\
        &=\left(\frac{d^2}{d^2-1}\right)^{d/2}
          \left(1-\frac2{d+1}\right)^{1/2}                                    \\
        &=\frac{d^d}{(d+1)(d^2-1)^{(d-1)/2}}.
\end{align*}
Taking logarithms,
\begin{align*}
        \log\frac{\vol(E^+)}{\vol(E)}
        &=d\log d-\log(d+1)-\frac{d-1}{2}\log(d^2-1)                         \\
        &=-\log\left(1+\frac1d\right)
          +\frac{d-1}{2}\log\left(1+\frac1{d^2-1}\right).
\end{align*}
The two elementary inequalities needed are
\[
        \log(1+t)\ge\frac{t}{1+t}\quad(t\ge0),
        \qquad
        \log(1+t)\le t\quad(t\ge0).
\]
Thus
\[
        \log\left(1+\frac1d\right)
        \ge\frac{1/d}{1+1/d}
        =\frac1{d+1},
\]
and
\[
        \frac{d-1}{2}\log\left(1+\frac1{d^2-1}\right)
        \le\frac{d-1}{2}\frac1{d^2-1}
        =\frac1{2(d+1)}.
\]
Therefore
\[
        \log\frac{\vol(E^+)}{\vol(E)}
        \le -\frac1{d+1}+\frac1{2(d+1)}
        =-\frac1{2(d+1)}.
\]
Transforming back gives
\[
        b=\frac{Qa}{\sqrt{a^{\mathsf T}Qa}},
        \qquad
        c^+=c-\frac1{d+1}b,
\]
\[
        Q^+=\frac{d^2}{d^2-1}\left(Q-\frac2{d+1}bb^{\mathsf T}\right).
\]
The singularity at $d=1$ is exactly why the one-dimensional case is handled by intervals instead of this formula.

\section*{Declaration of Generative AI and AI-Assisted Technologies in the Writing Process}
During the preparation of this work, the authors used DeepSeek to build a specialized agent for solving mathematical problems, which was employed to generate an initial proof of the main theorem. After using this tool, the authors reviewed and edited the content as needed and take full responsibility for the content of the published article.

\end{document}